\documentclass{amsart}
\usepackage{hyperref}
\usepackage{graphicx}
\usepackage{amsmath}
\usepackage{amssymb}
\usepackage{mathrsfs}
\usepackage{color}
\usepackage[utf8]{inputenc}

\newcommand*{\mailto}[1]{\href{mailto:#1}{\nolinkurl{#1}}}

\newtheorem{theorem}{Theorem}[section]
\newtheorem{definition}{Definition}[section]
\newtheorem{lemma}[theorem]{Lemma}

\newtheorem{remark}[theorem]{Remark}

\newcommand{\R}{{\mathbb R}}
\newcommand{\N}{{\mathbb N}}
\newcommand{\Z}{{\mathbb Z}}

\newcommand{\be}{\begin{equation}}
\newcommand{\ee}{\end{equation}}
\newcommand{\bea}{\begin{eqnarray}}
\newcommand{\eea}{\end{eqnarray}}

\newcommand{\I}{\mathrm{i}}
\newcommand{\E}{\mathrm{e}}

\newcommand{\re}{\mathrm{Re}}

\newcommand{\cn}{\mathrm{cn}}
\newcommand{\sn}{\mathrm{sn}}
\newcommand{\dn}{\mathrm{dn}}
\newcommand{\sech}{\mathrm{sech}}

\numberwithin{equation}{section}

\begin{document}

\title[Nonlinear differential identities for cnoidal waves]{Nonlinear differential identities for cnoidal waves}

\author[Michael Leitner]{Michael Leitner}
\address[Michael Leitner]{Forschungs-Neutronenquelle Heinz Maier-Leibnitz (FRM II)\\ Technische Universit\"at M\"unchen\\
Lichtenbergstrasse 1\\ 85748 Garching\\ Germany}
\email{michael.leitner@frm2.tum.de}

\author[Alice Mikikits-Leitner]{Alice Mikikits-Leitner}
\address[Alice Mikikits-Leitner]{Zentrum f\"ur Mathematik\\ Technische Universit\"at M\"unchen\\
Boltzmannstrasse 3 \\ 85748 Garching\\ Germany}
\email[Corresponding author]{mikikits@ma.tum.de}


\date{\today}

\begin{abstract}
This article presents a family of nonlinear differential identities for the spatially periodic function $u_s(x)$, which is essentially the Jacobian elliptic function $\cn^2(z;m(s))$ with one non-trivial parameter $s$. 
More precisely, we show that this function $u_s$ fulfills equations of the form
\begin{equation*}
\big(u_s^{(\alpha)}u_s^{(\beta)}\big)(x)=\sum_{n=0}^{2+\alpha+\beta}b_{\alpha,\beta}(n)u_s^{(n)}(x)+c_{\alpha,\beta},
\end{equation*}
for any $s>0$ and for all $\alpha,\beta\in\N_0$. 
We give explicit expressions for the coefficients $b_{\alpha,\beta}(n)$ and $c_{\alpha,\beta}$ for given $s$.

Moreover, we show that for any $s$ satisfying $\sinh(\pi/(2s))\geq 1$ the set of functions $\{1,u^{\vphantom{a}}_s,u'_s,u''_s,\dots\}$ constitutes a basis for $L^2(0,2\pi)$.
By virtue of our formulas the problem of finding a periodic solution to any nonlinear wave equation reduces to a problem in the coefficients. A finite ansatz exactly solves the KdV equation (giving the well-known cnoidal wave solution) and the Kawahara equation. An infinite ansatz is expected to be especially efficient if the equation to be solved can be considered a perturbation of the KdV equation.
\end{abstract}

\maketitle

\section{Introduction} \label{sec:intro}

Deriving solutions for integrable equations is a well-studied problem~\cite{gesztesy2003sea}. The used techniques stem from algebraic geometry (compact Riemann surfaces) and spectral analysis, relying heavily on the integrability of the underlying system. Historically the starting point was the development of the inverse scattering transform~\cite{gardner1974korteweg} and the Lax formalism~\cite{lax1968integrals} for the Korteweg--de Vries (KdV) equation, followed soon after by extensions to other integrable equations~\cite{ablowitz1974inverse, zakharov1974complete,zakharov1972exact,zakharov1979integration}.
These inverse spectral methods allow to deduce the existence of spatially localized solutions (solitons). A considerable breakthrough was made by extending the available tools to include periodic and certain classes of quasi-periodic solutions of integrable systems~\cite{dubrovin1975periodic, dubrovin1975inverse,its1975hill,its1975schroedinger,mckean1975spectrum}.

A significant simplification of the problem can be attained by considering only travelling-wave solutions, that is to transform the generic nonlinear 1+1-dimensional wave equation 
\begin{equation*}
N(v,v_t,v_x,\dots)=0
\end{equation*}
via a travelling wave ansatz 
\begin{equation*}
 v=f(\xi), \qquad \xi=kx-\omega t.
\end{equation*}
into a nonlinear ordinary differential equation. 

Generally, it has been found for integrable systems that solitary wave solutions take the form of hyperbolic functions, while their periodic travelling wave solutions can be written in terms of elliptic functions (e.g.~Jacobian elliptic functions $\cn$, $\sn$, $\dn$, Jacobian theta functions $\theta_i$ with $i\in\{0,1,2,3\}$, and the Weierstra\ss\ function $\wp$), from which the hyperbolic function solutions follow in the appropriate limit. In fact, these cnoidal wave solutions can actually be expressed as infinite sum of equally spaced solitary waves~\cite{whitham1984comments,boyd1984cnoidal}. 

The aim of this article is to propose a general ansatz for constructing periodic travelling-wave solutions of non-integrable nonlinear wave equations.
We will show that for a specific one-parameter family of functions $u_s$ the space spanned by the set
\begin{equation*}
 \{ 1, u_s, u_s', u_s'',\dots\}
\end{equation*}
is closed with respect to point-wise multiplication (i.e., it forms an algebra). Obviously, it is also closed with respect to differentiation. As a consequence, such an ansatz seems promising for solving ordinary differential equations where the nonlinear terms correspond to a polynomial of the respective derivatives. Probably the simplest non-trivial equation fulfilling these requirements is the KdV equation, and indeed our $u_s$ is nothing else than the periodic travelling wave solution of this equation, which is expressible in terms of the Jacobian elliptic function $\cn$. 

At the core of our method lie hitherto unknown nonlinear differential identities of elliptic functions
\begin{equation} \label{equ:introdiff}
\big(u_s^{(\alpha)}u_s^{(\beta)}\big)(x)=\sum_{n=0}^{2+\alpha+\beta}b_{\alpha,\beta}(n)u_s^{(n)}(x)+c_{\alpha,\beta},
\end{equation}
where $u_s^{(n)}\equiv d^nu_s/dx^n$ and $\alpha,\beta\in \N_0$. The coefficients $b_{\alpha,\beta}(n)$ and $c_{\alpha,\beta}$ depend on the parameter $s$ and are explicitly determined in terms of the Bernoulli numbers and infinite series of hyperbolic functions. Note that $u_s$ being a solution of the (integrated) KdV equation follows from \eqref{equ:introdiff} as a corollary by setting $\alpha=\beta=0$. That is, $v(x,t)=u(x-b_{0,0}(0)t)$ solves the KdV equation $v_t+2vv_x-b_{0,0}(2)v_{xxx}=0$, where the coefficients $b_{0,0}(0)$ and $b_{0,0}(2)$ will be specified later, see Section~\ref{subsec:specialcases}. Moreover, we use \eqref{equ:introdiff} also for obtaining an exact solution for the Kawahara equation, which is in general non-integrable. 

Additionally, we show that for $s$ small enough the closure of the vector space spanned by our proposed ansatz is actually the whole space of $2\pi$-periodic square integrable functions. Therefore, our approach is general and equivalent to a Fourier ansatz both in terms of solution set ($L^2(0,2\pi)$) and tools (Eqs.~\eqref{equ:introdiff}/trigonometric identities). Its great advantage lies in the additional freedom in the parameter $s$: In the limit of large $s$ the KdV solution $u_s$ evolves to a periodic sequence of separated solitary waves, and it can be expected that already a linear combination of its first few derivatives will constitute a very good (and in special cases even exact) approximation to the solution of a perturbation of the KdV equation, while for a Fourier ansatz the number of necessary terms would diverge. On the other hand, in the small $s$-limit $u_s$ essentially converges to $\cos$, and \eqref{equ:introdiff} correspond to the trigonometric identities (Remark~\ref{rem:trigon}).

We admit that a disadvantage of our proposed expansion compared to a Fourier ansatz lies in the fact that it is not orthogonal. However, by subjecting it to Gram-Schmidt-orthogonalization this could be easily remedied, with all above statements staying valid, at the expense of more complicated expressions for the coefficients $b_{\alpha,\beta}(n)$ and $c_{\alpha,\beta}$ in \eqref{equ:introdiff}.

Previous methods for solving general wave equations often employ a polynomial ansatz
\begin{equation} \label{equ:ansatzofothers}
 f(\xi)=\sum_{i=0}^N a_i u^i(\xi),
\end{equation}
with the building block $u(\xi)$ chosen as the explicitely known solution of some related equation. Here either a hyperbolic function such as $\tanh$~\cite{parkes1996automated,fan2000extended} or $\sech$~\cite{duffy1996travelling} for deriving solitary wave solutions, or an elliptic function such as $\sn$ and $\cn$~\cite{fu2001new,liu2001jacobi,peng2003exact} or $\wp$~\cite{kano1981exact,krishnan1984exact,kudryashov2004nonlinear} for periodic solutions have been used. 
Note, however, that this approach should be seen as purely heuristic, as it does not guarantee the expressibility of general solutions in the basis set, in contrast to our approach. Also, it can become quite tedious to express the resulting terms of products of derivatives by common linearly independent basis functions so that the coefficients can actually be compared. 

In the following we give a definition of our fundamental basis function $u_s$ and additional forms of representation.
\begin{definition} \label{def:us}
The $2\pi$-periodic function $u_s$ has the following equivalent representations: it can be expressed 
\begin{itemize}
\item[a)] as \emph{Fourier series}, i.e.,
\begin{equation} \label{equ:introsol}
u_s(x)=\sum_{k\in \Z} \frac{k}{\sinh(k\pi/s)}\E^{-\I kx},
\end{equation}
where the term corresponding to $k=0$ is to be understood as $s/\pi$ in the sense of de~l'Hospital,
\item[b)] as \emph{infinite sum of equally spaced solitary waves}, i.e.,
\begin{equation} \label{equ:solu2}
u_s(x)=\frac{s^2}{2}\sum_{n\in\Z} \sech^2\big( \frac{s}{2}(x-2\pi n)\big)
\end{equation}
or
\item[c)] in terms of \emph{elliptic functions}, i.e.,
\begin{equation} \label{equ:solu3}
u_s(x)=\frac{s}{\pi}+\big[\frac{2K(m)^2(1-m)}{\pi^2} -\frac{2K(m)E(m)}{\pi^2}\big]-\frac{2 mK(m)^2}{\pi^2}\cn^2\big( \frac{K(m)}{\pi}x;m\big),
\end{equation}
where $\cn(z;m)$ is a Jacobian elliptic function. Here $m$ denotes the elliptic modulus related to the parameter $s$ via the relation 
\begin{equation} \label{equ:relationsm}
s=\frac{K(m)}{K(1-m)},
\end{equation}
and $K(m)$ denotes the complete elliptic integral of the first kind, i.e.~$K(m)=\int_0^{\pi/2}ds/\sqrt{1-m\sin s}$. Moreover, $E(m)$ is the complete elliptic integral of the second kind, i.e.~$E(m)=\int_0^{\pi/2}ds\sqrt{1-m\sin s}$.
\end{itemize}
\end{definition}

A periodic solution of the KdV equation in terms of the Jacobian elliptic function $\cn$, see~\eqref{equ:solu3}, has been already derived by Korteweg and de~Vries~\cite{korteweg1895change} and is known as \emph{cnoidal wave}, having three independent parameters. Later Whitham~\cite{whitham1984comments} and Boyd~\cite{boyd1984cnoidal} derived independently that a KdV cnoidal wave solution can be expressed as infinite sum of equally spaced solitary waves, see~\eqref{equ:solu2}. While Whitham used the method of partial fraction decomposition for hyperbolic functions to directly verify that this is a KdV solution, Boyd just made use of the Fourier series of the elliptic function $\cn$, see~\eqref{equ:introsol}, and subsequently applied Poisson summation as detailed in Appendix~\ref{app:poissonsum} to get the representation~\eqref{equ:solu2}.

To be specific, consider the function $u_s$ given by~\eqref{equ:introsol}. The representation~\eqref{equ:solu2} follows from~\eqref{equ:poissonsum} with $T=2\pi$ and
\begin{equation*}
F(\omega)=\frac{\omega}{\sinh(\omega \pi/s)},\qquad \mathcal{F}^{-1}(F)(x)=\frac{s^2}{4\pi}\sech^2\big( \frac{sx}{2}\big).
\end{equation*}
To show statement c) use the Fourier series of the elliptic function $\dn^2$, see~\cite[p.~25]{oberhettinger1973methods},
\begin{equation*}
\frac{K^2(m)}{2\pi^2}\dn^2\big( \frac{K(m)}{\pi}x;m\big)-\frac{K(m)E(m)}{2\pi^2}=\sum_{k=1}^{\infty} \frac{kq^k}{1-q^{2k}}\cos(kx),\quad q=\E^{-\pi K(1-m)/K(m)},
\end{equation*}
Setting $s=K(m)/K(1-m)$ and noting that $\dn^2=(1-m)+m\,\cn^2$, see e.g.~\cite[p.~20]{byrd1971handbook}, then yields~\eqref{equ:solu3}.

\section{Main results} \label{sec:results}
The main purpose of this section is to present a family of identities that allow to write any product of derivatives of the function  $u_s$ given by Definition~\ref{def:us} as a finite linear combination of its derivatives.
The main effort will lie in computing the discrete convolution of specific sequences (see~Theorem~\ref{thm:result1}), corresponding to the Fourier coefficients of arbitrary derivatives of $u_s$. These will allow us to directly obtain the differential identities of Theorem~\ref{thm:result3}. 
Finally, we show that for any parameter $s$ satisfying the constraint $\sinh(\pi/(2s))\geq 1$ the functions $u_s^{(n)}(x)$ with $n\in \N_0$ together with the identity function actually constitute a basis of $L^2(0,2\pi)$.
The usefulness of our results in the context of nonlinear wave equations will be demonstrated in Section~\ref{sec:appl}.

\begin{theorem}[{\bf Discrete convolution formulas}] \label{thm:result1}
Let $\alpha,\beta\in \N_0$ and $s\in (0,\infty)$. Then we have
\begin{equation} \label{equ:polynomial}
\sum_{k\in \Z} \frac{k^{1+\alpha}(k+j)^{1+\beta}}{\sinh (k\pi /s)\sinh ((k+j)\pi /s)}= \sum_{n=0}^{2+\alpha+\beta} {a}_{\alpha, \beta}(n) \frac{j^{1+n}}{\sinh (j\pi /s)},\quad j\neq 0.
\end{equation}
In the case $n=2+\alpha+\beta$ the coefficients ${a}_{\alpha, \beta}(n)$ are given by
\begin{equation} \label{equ:fouriercoeff1a}
a_{\alpha, \beta}(2+\alpha+\beta)=
2(-1)^{\alpha}\frac{(1+\alpha)!(1+\beta)!}{(3+\alpha+\beta)!},
\end{equation}
the term for $n=1+\alpha+\beta$ vanishes, and for $n\leq \alpha+\beta$ we have
\begin{align} \label{equ:fouriercoeff2a}
a_{\alpha, \beta}(n)=&
(-1)^{1+\alpha+\beta}\Big( \binom{1+\alpha}{1+n} +(-1)^{n}\binom{1+\beta}{1+n}\Big)e_{2+\alpha+\beta-n}\\
&+(-1)^{\alpha}\frac{s}{\pi}\delta_{\alpha,n}\delta_{\beta,0}+\frac{s}{\pi}\delta_{\beta,n}\delta_{\alpha,0}\nonumber
\end{align}
where $\delta_{i,j}$ denotes the Kronecker delta, i.e.~$\delta_{i,j}=1$ for $i=j$ and $\delta_{i,j}=0$ else, and the quantities $e_{\ell}$ ($\ell \geq 2$) are defined by
\begin{equation} \label{equ:eell}
e_{\ell}=(1+(-1)^{\ell})\Big(\frac{B_{\ell}}{\ell}+2\sum_{k=1}^{\infty} \frac{k^{\ell-1}}{1-\E^{2\pi k/s}}\Big),
\end{equation}
where $B_{\ell}$ denotes the ${\ell}$th Bernoulli number, see Remark~\ref{rem:bernoullinumbers}.
\end{theorem}
The quantities $a_{\alpha,\beta}(n)$ given by~\eqref{equ:fouriercoeff2a} are zero if $\alpha+\beta+n$ is odd and they fulfill the property
\begin{equation} \label{equ:symmetrycoeffa}
a_{\alpha, \beta}(n)=(-1)^{\alpha+\beta}a_{\beta,\alpha}(n).
\end{equation}

We think that besides being preparatory work for Theorem~\ref{thm:result3}, the discrete convolution formulas given in Theorem~\ref{thm:result1} are interesting in their own right. We found no such general result in standard tables~\cite{gradstejn1980table,hansen1975table}. 
Note, however, that in a letter to Hardy \cite[p.~XXVI,~VI~(8)]{ramanujan1927collected} Ramanujan stated the related formula
\begin{equation} \label{equ:ramanujan}
\sum_{k=1}^{\infty}\frac{k^{4m}}{\sinh^2(k\pi)}=-\frac{1}{2\pi}B_{4m}-\frac{4m}{\pi}\sum_{k=1}^{\infty}\frac{k^{4m-1}}{1-\E^{2\pi k}},\qquad m\in\N,
\end{equation}
corresponding to the special case of $s=1$ but with $j=0$, being outside the range of validity of Theorem~\ref{thm:result1}. This formula was first proven by Preece~\cite{preece1928theorems} using the residue theorem. 

\begin{remark} \label{rem:bernoullinumbers}
The (first) Bernoulli numbers are defined recursively via, see \cite{abramowitz1964handbook},
\begin{equation*}
B_0=1, \quad B_1=-\frac{1}{2}, \quad B_{\ell}=-\sum_{k=0}^{{\ell}-1}\binom{{\ell}}{k}\frac{B_k}{{\ell}-k+1} \quad \text{for ${\ell}\geq 2$}.
\end{equation*}
Hence
\begin{equation*}
B_0=1, \quad B_1=-1/2,\quad B_2=1/6, \quad B_4=-1/30, \quad B_6=1/42, \quad B_8=-1/30, \dots,
\end{equation*}
while the Bernoulli numbers with odd index larger than one vanish, i.e.~$B_{2n+1}=0$ for $n\in \N$.

Additionally, one defines the second Bernoulli numbers by $B_{\ell}^*=B_{\ell}$ for ${\ell}\neq1$ and $B_1^*=-B_1$.
\end{remark}

The following theorem presents a family of identities involving products of derivatives of the elliptic function $u_s$.
\begin{theorem}[{\bf Nonlinear differential identities}] \label{thm:result3}
Let $\alpha,\beta\in \N_0$, $s\in (0,\infty)$, and $u_s$ be the $2\pi$-periodic function given by Definition~\ref{def:us}.
Then the nonlinear differential equations
\begin{equation} \label{equ:nonlinpde}
\Big(u_s^{(\alpha)}u_s^{(\beta)}\Big)(x)=\sum_{n=0}^{2+\alpha+\beta}b_{\alpha,\beta}(n)u_s^{(n)}(x)+c_{\alpha,\beta},
\end{equation}
where $u_s^{(n)}\equiv d^nu_s/dx^n$, are fulfilled. The coefficients ${b}_{\alpha,\beta}(n)$ are given by
\begin{equation}
 {b}_{\alpha,\beta}(n)=(-1)^{\frac{\alpha-\beta+n}{2}}{a}_{\alpha,\beta}(n),
\end{equation}
with ${a}_{\alpha,\beta}(n)$ determined by~\eqref{equ:fouriercoeff1a} and~\eqref{equ:fouriercoeff2a}.
That is, in the case $n=2+\alpha+\beta$ we have
\begin{equation} \label{equ:fouriercoeff1}
b_{\alpha, \beta}(2+\alpha+\beta)=-2\frac{(1+\alpha)!(1+\beta)!}{(3+\alpha+\beta)!},
\end{equation}
the term for $n=1+\alpha+\beta$ vanishes, and for $n\leq \alpha+\beta$ we have
\begin{align} \label{equ:fouriercoeff2}
b_{\alpha, \beta}(n)=&
(-1)^{1+\frac{\beta-\alpha+n}{2}}\Big( \binom{1+\alpha}{1+n} +(-1)^{n}\binom{1+\beta}{1+n}\Big)e_{2+\alpha+\beta-n}\\
&+\frac{s}{\pi}\delta_{\alpha,n}\delta_{\beta,0}+\frac{s}{\pi}\delta_{\beta,n}\delta_{\alpha,0}\nonumber
\end{align}
where $\delta_{i,j}$ denotes the Kronecker delta, i.e.~$\delta_{i,j}=1$ for $i=j$ and $\delta_{i,j}=0$ else, and the quantities $e_{\ell}$ ($\ell \geq 2$) are defined by~\eqref{equ:eell}.
The constant term in~\eqref{equ:nonlinpde} is given by
\begin{equation} \label{equ:constcoeff}
c_{\alpha,\beta}=(-1)^{\frac{\alpha-\beta}{2}}F_{\alpha+\beta}-\frac{s}{\pi}{b}_{\alpha,\beta}(0),
\end{equation}
where $F_{\alpha+\beta}$ denotes the infinite sum
\begin{equation} \label{equ:infsumF}
F_{\alpha+\beta}=\sum_{k\in \Z} \frac{k^{2+\alpha+\beta}}{\sinh^2(k\pi/s)}.
\end{equation}
\end{theorem}
Note that the infinite sum $F_{\alpha+\beta}$ and the coefficient ${b}_{\alpha,\beta}(0)$ are zero if $\alpha+\beta$ is odd. Hence, in this case, the constant term in~\eqref{equ:nonlinpde} vanishes, i.e.
\begin{equation*}
 c_{\alpha,\beta}=0 \qquad \text{if $\alpha+\beta$ odd}.
\end{equation*}
Moreover, the coefficients fulfill the symmetry property
\begin{equation} \label{equ:symmetrycoeff}
b_{\alpha,\beta}(n)=b_{\beta,\alpha}(n).
\end{equation}

\begin{remark} \label{rem:trigon}
The function $u_s$ given by~\eqref{equ:solu3} depends on one free parameter $s>0$.
In the small amplitude limit, which corresponds to the limit $s\to 0$ or equivalently $m\to 0$, see~\eqref{equ:relationsm}, we have $\cn(z;m(s))\to\cos(z)$~\cite[eq.~16.13.2]{abramowitz1964handbook}. This indicates that the identities~\eqref{equ:nonlinpde} can be considered generalizations of the trigonometric identities for products of cosine and sine to the elliptic case. To demonstrate this consider the first equation given in~\eqref{equ:odes}. Inserting the Fourier representation~\eqref{equ:introsol} for $u_s$ and the expressions~\eqref{equ:esmalls} and \eqref{equ:infsumFsmall} for $e_2$ (turning up in the coefficient $b_{0,0}(0)$) and $F_0$, respectively, the leading order terms, i.e.~terms of order $\exp(-2\pi/s)$, give the well-known formula $2\cos^2(x)= 1+\cos(2x)$. Analogously, the second and third equation in~\eqref{equ:odes} asymptotically give the identities $2\cos(x)\sin(x)=\sin(2x)$ and $2\sin^2(x)=1-\cos(2x)$, respectively.

On the other hand, in the limit $s\to \infty$ (corresponding to $m\to 1$) the function $u_s$ becomes a sequence of independent solitary waves given by $s^2\sech^2(sx/2)/2$ (use~\cite[eq.~16.15.2]{abramowitz1964handbook}). Then the first equation in our system~\eqref{equ:nonlinpde} just states that this is a solution of the KdV equation, known as $1$-soliton.
\end{remark}

\begin{theorem}[{\bf Basis for $L^2(0,2\pi)$}] \label{thm:result4}
Let $s>0$ be such that $\sinh(\pi/(2s))\geq 1$ and $u_s$ the $2\pi$-periodic function given by Definition~\ref{def:us}. Then the set of functions
\begin{equation} \label{equ:basis}
 \{1,u^{\vphantom{a}}_s,u'_s,u''_s,\dots\}
\end{equation}
constitutes a basis for $L^2(0,2\pi)$, i.e.~the space of square-integrable $2\pi$-periodic functions.
\end{theorem}

The condition $\sinh(\pi/(2s))\geq 1$ is fulfilled for approximately $s< 1.782$. By~\eqref{equ:relationsm} this corresponds to an elliptic modulus $m<0.943$ in the elliptic function representation~\eqref{equ:solu3} of $u_s$. Note that our proof does not rule out the possibility of our function set being a basis also for larger $s$. However, as for very large $s$ the function $u_s$ consists of spikes at multiples of $2\pi$ and is essentially flat in-between, it is very plausible that there exists some finite $s$ above which it becomes impossible to describe arbitrary features at $\pi$ while controlling the behaviour at 0.

\section{Proofs} \label{sec:proof}
The aim of this section is to prove the results presented in Section~\ref{sec:results}. For Theorem~\ref{thm:result3} we have to show that the Fourier series of $u_s$ given by~\eqref{equ:introsol} solves~\eqref{equ:nonlinpde} for arbitrary $\alpha,\beta\in \N_0$, where the coefficients $b_{\alpha,\beta}(n)$, $n=0,\dots,2+\alpha+\beta$, and $c_{\alpha,\beta}$ are determined by~\eqref{equ:fouriercoeff1}--\eqref{equ:constcoeff}. The idea will be to show that the Fourier coefficients of $u_s$ fulfill the Fourier transformed equation. This is exactly the statement of Theorem~\ref{thm:result1}. Thus, the starting point will be to calculate the sum over $k\in \Z$ of the function $f_{\alpha,\beta}:\Z \to \R$ defined by
\begin{equation} \label{equ:mainfunction2}
f_{\alpha, \beta}(k)=\frac{k^{1+\alpha}(k+j)^{1+\beta}}{\sinh \big(k\pi /s\big)\sinh \big((k+j)\pi /s\big)}\cdot \frac{\sinh (j\pi /s)}{j}, \qquad j\neq 0.
\end{equation}
The function $f_{\alpha,\beta}(k)$ has two singularities at the points $k=-j$ and $k=0$.
To simplify notation, we will separate the set of integers into two sets, namely those lying to the \emph{left} respectively to the \emph{right} of the point $k=-j/2$. Additionally, we split each of these two sets into two further sets, lying in the \emph{inside}, that is between $-j$ and $0$, or in the \emph{outside}. This procedure results in the following four distinct sets
\begin{align*}
L_O&=\{ k\in \Z | k\leq -j\},\\
L_I&=\{ k\in \Z | -j<k\leq -j/2\},\\
R_I&=\{ k\in \Z | -j/2<k<0\},\\
R_O&=\{ k\in \Z | k\geq 0\},
\end{align*}
where we assumed without loss of generality $j>0$. If $j$ is even, the central point $-j/2$ needs special consideration. Here we follow the arbitrary choice to include it in $L_I$. Moreover, it will be useful to define the following sets
\begin{equation*}
L=L_O\cup L_I, \quad R=R_O\cup R_I, \quad I=L_I\cup R_I, \quad O=L_O\cup R_O.
\end{equation*}

After settling the notation we now come to prove Theorem~\ref{thm:result1}.

\subsection{Proof of Theorem~\ref{thm:result1}}

\begin{proof}
Obviously, the denominator of $f_{\alpha, \beta}(k)$ for large $j$ is essentially constant between $-j$ and $0$ and decays exponentially on the outside. Its zeros will be cancelled by the corresponding zeros of the numerator. In a coarse view, the function will therefore correspond to a polynomial on $I$ and vanish on $O$. To be specific, define the functions
\begin{equation*}
\tilde{g}_{\alpha,\beta}(k)=-\frac{2}{j}k^{1+\alpha}(k+j)^{1+\beta}
\end{equation*}
(note that $j\neq 0$ per the conditions of the theorem) and
\begin{equation}
g_{\alpha, \beta}(k)=\left\{ \begin{matrix} \tilde{g}_{\alpha,\beta}(k),& \quad k\in I,\\ 0,& \quad k\in O. \end{matrix}\right.
\end{equation}

As motivated above, for large $j$ the function  $g_{\alpha, \beta}$ will approach $f_{\alpha, \beta}$ away from the singularities at $-j$ and $0$. The idea of the proof will now be to split the computation of the sum into a sum over $g_{\alpha, \beta}$ on $I$ and a sum over the residuals, where the major contributions will come from the regions around the singularities:
\begin{equation} \label{equ:poldecomp}
\sum_{k\in \Z} f_{\alpha, \beta}(k)=S_1+S_2,\text{ with }S_1=\sum_{k\in \Z} \big(f_{\alpha, \beta}-g_{\alpha, \beta}\big)(k), \text{ and }S_2=\sum_{k\in \Z} g_{\alpha, \beta}(k)
\end{equation}

\noindent {\it Step 1: Computation of $S_1$}

We define
\begin{equation} \label{equ:defh}
\begin{aligned}
& h_1(k)=\tilde{g}_{\alpha,\beta}(k)\frac{1}{1-\E^{-2(k+j)\pi /s}}, \\
& h_2(k)=\tilde{g}_{\alpha,\beta}(k) \frac{-1}{1-\E^{2(k+j)\pi /s}}, \\
& h_3(k)=\tilde{g}_{\alpha,\beta}(k) \frac{-1}{1-\E^{-2k\pi /s}}, \\
& h_4(k)=\tilde{g}_{\alpha,\beta}(k)\frac{1}{1-\E^{2k\pi /s}},
\end{aligned}
\end{equation}
and
\begin{equation} \label{equ:defhdelta}
\Delta h_1(k)=\Delta h_2(k)=h_3(k), \qquad \Delta h_3(k)=\Delta h_4(k)=h_2(k),
\end{equation}
resulting in
\begin{equation*}
\big(f_{\alpha, \beta}-g_{\alpha, \beta}\big)(k)=\left\{ 
\begin{matrix} 
h_1+\Delta h_1, & k\in L_O,\\
h_2+\Delta h_2, & k\in L_I,\\
h_3+\Delta h_3, & k\in R_I,\\
h_4+\Delta h_4, & k\in R_O.\\
\end{matrix}\right.
\end{equation*}

For the purpose of computing the first sum in~\eqref{equ:poldecomp} we note that the identity
\begin{equation*}
S_1= \sum_{k\in L_O}(h_1+\Delta h_1)+\sum_{k\in L_I}(h_2+\Delta h_2)+
\sum_{k\in R_I}(h_3+\Delta h_3)+\sum_{k\in R_O}(h_4+\Delta h_4)
\end{equation*}
can be transformed by using~\eqref{equ:defhdelta} and thus equivalently written as
\begin{equation} \label{equ:sum1}
S_1=\sum_{k\in L_O}h_1+\sum_{k\in L_I\cup R}h_2+\sum_{k\in R_I\cup L}h_3+\sum_{k\in R_O}h_4,
\end{equation}
We will split this sum according to $S_1=S_{1a}+S_{1b}$, then we compute
\begin{align*}
S_{1a}=&\sum_{k\in L_O}h_1(k)+\sum_{k\in L_I\cup R}h_2(k)=\sum_{k=-\infty}^{-1}h_1(k-j)+h_1(-j)+\sum_{k=1}^{\infty}h_2(k-j)\\
=& h_1(-j)-\frac{2}{j}\left\{ \sum_{k=-\infty}^{-1}\frac{(k-j)^{1+\alpha}k^{1+\beta}}{1-\E^{-2k\pi /s}} +
\sum_{k=1}^{\infty}\frac{-(k-j)^{1+\alpha}k^{1+\beta}}{1-\E^{2k\pi /s}}\right\}\\
=& h_1(-j)-\frac{2}{j} \sum_{k=1}^{\infty} \frac{-k^{1+\beta}\big( (k-j)^{1+\alpha}+(-1)^{1+\alpha+\beta}(k+j)^{1+\alpha}\big)}{1-\E^{2k\pi /s}}\\
=& h_1(-j) -2\sum_{i=0}^{1+\alpha} \left\{ (-1)^{\alpha+\beta}[1+(-1)^{1+\alpha+\beta-i}]\binom{1+\alpha}{i} \sum_{k=1}^{\infty} \frac{k^{2+\alpha+\beta-i}}{1-\E^{2k\pi /s}}\right\}j^{i-1},
\end{align*}
with $h_1(-j)=\delta_{\beta,0}(-1)^{\alpha}j^{\alpha}s/\pi$ due to the rule of de~l'Hospital.
Analogously, we have
\begin{align*}
S_{1b}=&\sum_{k\in R_I\cup L}h_3+\sum_{k\in R_O}h_4=\sum_{k=-\infty}^{-1}h_3(k)+h_4(0)+\sum_{k=1}^{\infty}h_4(k)\\
=& h_4(0) -2\sum_{i=0}^{1+\beta} \left\{ [1+(-1)^{1+\alpha+\beta-i}]\binom{1+\beta}{i} \sum_{k=1}^{\infty} \frac{k^{2+\alpha+\beta-i}}{1-\E^{2k\pi /s}}\right\}j^{i-1},
\end{align*}
with $h_4(0)=\delta_{\alpha,0}j^{\beta}s/\pi$.
Thus, the sum $S_1$ given by \eqref{equ:sum1} is equal to
\begin{equation} \label{equ:step1}
\begin{aligned}
S_1=&(-1)^{\alpha}j^{\alpha}\frac{s}{\pi}\delta_{\beta,0}+j^{\beta}\frac{s}{\pi}\delta_{\alpha,0}\\
&-2 \sum_{n=0}^{\infty} [1+(-1)^{\alpha+\beta-n}]\left\{ \binom{1+\beta}{1+n}+(-1)^{\alpha+\beta}\binom{1+\alpha}{1+n}\right\}\sum_{k=1}^{\infty} \frac{k^{1+\alpha+\beta-n}}{1-\E^{2k\pi /s}}j^{n}.
\end{aligned}
\end{equation}
\vspace{4mm}

\noindent {\it Step 2: Computation of $S_2$}

We now compute the second sum in~\eqref{equ:poldecomp}. We have
\begin{align*}
S_2&=\sum_{k\in I} \tilde{g}_{\alpha,\beta}(k)=
-\frac{2}{j}\sum_{k=-j}^{0} k^{1+\alpha}(k+j)^{1+\beta}\\
&=-\frac{2}{j}\sum_{k=1}^j \sum_{i=0}^{1+\beta} (-1)^{\alpha+\beta-i} \binom{1+\beta}{i}k^{2+\alpha+\beta-i}j^i\\
&= 2\sum_{i=0}^{1+\beta} \sum_{n=0}^{2+\alpha+\beta-i} \frac{(-1)^{1+\alpha+\beta-i} }{3+\alpha+\beta-i} \binom{1+\beta}{i} \binom{3+\alpha+\beta-i}{n} B^*_n j^{2+\alpha+\beta-n}\\
&=2\sum_{n=0}^{2+\alpha+\beta} \sum_{i=0}^{\min\{1+\beta,n\}} \frac{(-1)^{1+\alpha+\beta-i} }{3+\alpha+\beta-i} \binom{1+\beta}{i} \binom{3+\alpha+\beta-i}{2+\alpha+\beta-n} B^*_{2+\alpha+\beta-n} j^{n},
\end{align*}
where we have used the identity, see~\cite[Section~0.12]{gradstejn1980table},
\begin{equation*}
\sum_{k=1}^j k^q=\frac{1}{q+1}\sum_{n=0}^{q} \binom{q+1}{n} B^*_n j^{q+1-n},
\end{equation*}
where $B^*_n$ denotes the $n$th second Bernoulli number (see~Remark~\ref{rem:bernoullinumbers}).

Next, let us define for $\alpha,\beta,n\in\N$ with $n\leq 2+\alpha+\beta$ the following function
\begin{equation*}
s(\alpha,\beta,i,n)=\frac{(-1)^{1+\alpha+\beta-i} }{3+\alpha+\beta-i} \binom{1+\beta}{i} \binom{3+\alpha+\beta-i}{2+\alpha+\beta-n},
\end{equation*}
such that
\begin{equation*}
 S_2=2\sum_{n=0}^{2+\alpha+\beta} \sum_{i=0}^{\min\{1+\beta,n\}}s(\alpha,\beta,i,n) B^*_{2+\alpha+\beta-n} j^{n}.
\end{equation*}
For $i\le\min\{1+n,1+\beta\}$ we have 
\begin{equation*}
 \frac{s(\alpha,\beta,i+1,n)}{s(\alpha,\beta,i,n)}=\frac{(i-(1+n))(i-(1+\beta))}{(i-(2+\alpha+\beta))}\frac{1}{(i+1)},
\end{equation*}
while $s(\alpha,\beta,i,n)=0$ for $i=1+\min\{1+n,1+\beta\}$. This being able to write the quotient of successive summands by a expression that is rational in the summation index is the defining criterion for hypergeometric series. To be specific, we have
\begin{equation} \label{equ:sum1plus}
S^*=\sum_{i=0}^{1+\min\{n,\beta\}} s(\alpha,\beta,i,n)= \frac{(-1)^{1+\alpha+\beta} }{3+\alpha+\beta} \binom{3+\alpha+\beta}{2+\alpha+\beta-n} F\left[ \begin{array}{c} -1-n, -1-\beta \\-2-\alpha-\beta \end{array};1\right],
\end{equation}
where $F(a,b;c;z)$ denotes the Gaussian hypergeometric series defined by (see e.g.~\cite[eq.~15.1.1]{abramowitz1964handbook}) 
\begin{equation} \label{equ:gauss}
F\left[ \begin{array}{c} a, b \\c \end{array};z\right]=\sum_{k=0}^{\infty}\frac{(a)_k(b)_k}{(c)_k}\frac{z^k}{k!},
\end{equation}
and $(\alpha)_k$ is the Pochhammer symbol given by
\begin{equation*}
(\alpha)_k=\alpha(\alpha+1)\dots(\alpha+k-1)\quad \text{for }k=1,2,\dots,\quad (\alpha)_0=1.
\end{equation*}

Note that in the case of negative integer $a$ or $b$, the expression \eqref{equ:gauss} is to be understood as the finite series up to the first summand equal to zero, and is therefore a polynomial. In this case and for $z=1$ we have
\begin{equation} \label{equ:gaussthm}
 F\left[ \begin{array}{c} a, b \\c \end{array};1\right]=\frac{(c-a)_{-b}}{(c)_{-b}},
\end{equation}
which holds as long as the series converges\footnote{This is the case if $-b\in\N$ and $c\not\in\{b,b+1\dots 0\}$ (or equivalently for $a$), and follows from Gauss' theorem (that is $F(a,b;c;1)=\Gamma(c)\Gamma(c-a-b)/\Gamma(c-a)\Gamma(c-b)$ if $\re (c-a-b)>0$ and $c$ no negative integer, see~\cite[eq.~15.1.20]{abramowitz1964handbook}): Multiply \eqref{equ:gaussthm} by $(c)_{-b}$, then for fixed $b$ and $a$ both sides are polynomials in $c$ of degree $b$ that agree on infinitely many values and are therefore equal everywhere.}.
Hence \eqref{equ:sum1plus} becomes 
\begin{equation*}
S^*= \frac{(-1)^{1+\alpha+\beta}}{3+\alpha+\beta} \binom{3+\alpha+\beta}{2+\alpha+\beta-n} \frac{(-1-\alpha-\beta+n)_{1+\beta}}{(-2-\alpha-\beta)_{1+\beta}}.
\end{equation*}
In the case $n\leq 1+\alpha+\beta$ we have
\begin{equation*}
S^*=\frac{(-1)^{1+\alpha+\beta}}{2+\alpha+\beta-n}\binom{1+\alpha}{1+n},
\end{equation*}
while for $n=2+\alpha+\beta$ we compute
\begin{equation*}
S^*=\frac{(-1)^{1+\alpha+\beta}}{3+\alpha+\beta} \frac{(1)_{1+\beta}}{(-2-\alpha-\beta)_{1+\beta}}=(-1)^{\alpha}\frac{(1+\alpha)!(1+\beta)!}{(3+\alpha+\beta)!}
\end{equation*}
Moreover, we have
\begin{align*}
\ s(\alpha,\beta,1+n,n)=\frac{(-1)^{\alpha+\beta-n}}{2+\alpha+\beta-n}\binom{1+\beta}{1+n},
\end{align*}
which is nonzero only if $n\le \beta$. 
Thus, finally, step~2 yields
\begin{equation} \label{equ:step2}
\begin{aligned}
 S_2=&2(-1)^{\alpha}\frac{(1+\alpha)!(1+\beta)!}{(3+\alpha+\beta)!}\\
&+2\sum_{n=0}^{1+\alpha+\beta} \frac{(-1)^{1+\alpha+\beta}}{2+\alpha+\beta-n}\Big( \binom{1+\alpha}{1+n}+(-1)^n\binom{1+\beta}{1+n}\Big) B^*_{2+\alpha+\beta-n}j^n.
\end{aligned}
\end{equation}
Note that the term for $n=1+\alpha+\beta$ vanishes. Thus we can replace $B^*_{2+\alpha+\beta-n}$ by $B_{2+\alpha+\beta-n}$, see Remark~\ref{rem:bernoullinumbers}.

Summarizing step~1 and step~2 by inserting~\eqref{equ:step1} and~\eqref{equ:step2} into~\eqref{equ:poldecomp} we deduce
\begin{equation*} 
\sum_{k\in\Z}f_{\alpha, \beta}(k)= \sum_{n=0}^{2+\alpha+\beta} {a}_{\alpha, \beta}(n) j^n \quad \text{for}\quad j\neq 0,
\end{equation*}
where the coefficients ${a}_{\alpha, \beta}(n)$ are given by~\eqref{equ:fouriercoeff1a}--\eqref{equ:fouriercoeff2a}.
\end{proof}

\subsection{Proof of Theorem~\ref{thm:result3}}

\begin{proof}
For $u_s$ given by~\eqref{equ:introsol} we have
\begin{equation*}
u_s^{(n)}(x)=\I^n\sum_{k\in \Z} \frac{k^{1+n}}{\sinh (k\pi /s)}\E^{\I kx}.
\end{equation*}
We compute
\begin{align*}
\big( u_s^{(\alpha)}u_s^{\beta)}\big)(x)&=\I^{\alpha+\beta}\Big(\sum_{k\in \Z}\frac{k^{1+\alpha}}{\sinh (k\pi /s)}\E^{\I kx}\Big)\Big(\sum_{\ell\in \Z}\frac{\ell^{1+\beta}}{\sinh (\ell\pi /s)}\E^{\I \ell x}\Big)\\
&=\I^{\alpha-\beta}\sum_{j\in \Z}\sum_{k\in\Z}\frac{k^{1+\alpha}(k+j)^{1+\beta}}{\sinh (k\pi /s)\sinh((k+j)\pi /s)}\E^{-\I jx}\\
&= \I^{\alpha-\beta}\sum_{j\in\Z\setminus \{0\}}\sum_{k\in\Z}f_{\alpha,\beta}(k)\frac{j}{\sinh(j\pi /s)}\E^{-\I jx}+\I^{\alpha-\beta} F_{\alpha+\beta},
\end{align*}
where we substituted $\ell=-k-j$ and $f_{\alpha,\beta}(k)$ is given by~\eqref{equ:mainfunction2}. Hence by using Theorem~\ref{thm:result1} we get
\begin{align*}
\big( u_s^{(\alpha)}u_s^{\beta)}\big)(x)&= \I^{\alpha-\beta} \sum_{n=0}^{2+\alpha+\beta}{a}_{\alpha,\beta}(n)\sum_{j\in\Z\setminus \{0\}}\frac{j^{1+n}}{\sinh(j\pi /s)}\E^{-\I jx}+\I^{\alpha-\beta} F_{\alpha+\beta}\\
&= -\I^{\alpha-\beta}\frac{s}{\pi}{a}_{\alpha,\beta}(0)+\sum_{n=0}^{2+\alpha+\beta}\I^{\alpha-\beta+n}{a}_{\alpha,\beta}(n)u_s^{(n)}(x)+\I^{\alpha-\beta} F_{\alpha+\beta},
\end{align*}
where $F_{\alpha+\beta}$ is defined by~\eqref{equ:infsumF}. Finally, setting ${b}_{\alpha,\beta}(n)=(-1)^{\frac{\alpha-\beta+n}{2}}{a}_{\alpha,\beta}(n)$ finishes the proof.
\end{proof}

\subsection{Proof of Theorem~\ref{thm:result4}}

\begin{proof}
The set of functions $\{\E^{\I jx}\}_{j\in \Z}$ forms an (orthogonal) basis for $L^2(0,2\pi)$. Due to Plancherel's theorem~\cite[p.~17]{stein1975introduction} this carries over to the space of Fourier coefficients. Indeed, the Fourier transform of these basis functions
\begin{equation*}
 \mathcal{F}(\E^{\I jx})(k)=F_j(k)=\delta_{k,j},
\end{equation*}
where $\delta_{k,j}$ denotes the Kronecker delta, is obviously a basis for $l^2(\Z)$.
Note that the Fourier transform of the functions $u_s^{(m)}(x)$ is
\begin{equation*}
\mathcal{F}\big( u_s^{(m)}\big)(k)=U_m(k)=\frac{k^{1+m}}{\sinh(\lambda k)}, \quad \text{with $\lambda=\frac{\pi}{s}$},
\end{equation*}
see the Fourier series representation~\eqref{equ:introsol} of $u_s$. Since $\{k^{\ell}\}_{\ell\geq 1}$ are linearly independent so are $\{U_m\}_{m\geq 0}$.

The idea of the proof will be to show that any element of the Fourier basis can be described by a finite linear combination of the derivatives of $u_s$ to arbitrary accuracy. We will give an explicit expression for this description, motived by the idea of Lagrange interpolation of the Fourier components. 

To be specific, consider a fixed $j\in\Z$. The goal is to find a linear combination of the derivatives of $u_s$ and the constant function so that its Fourier transform approaches $\delta_{j,k}$. If $j=0$, then this linear combination will only involve the constant function. Without loss of generality we can therefore consider $j\ne 0$ and need to describe $\delta_{j,k}$ only for $k\ne 0$ by the Fourier transforms $U_m(k)$ (the constant component taking care of $k=0$). 

For a given $n\in\N$ with $n\ge|j|$ the expression
\begin{equation} \label{equ:interpoly}
f_j^n(k)=\frac{\sinh(\lambda j)}{\sinh(\lambda k)}\prod_{\substack{i=-n\\i\neq 0, j}}^n\frac{k-i}{j-i}
\end{equation}
obviously is a linear combination of the Fourier transforms $U_m$, $0\le m\le 2n-2$ and fulfills $f_j^n(k)=\delta_{k,j}$ for $0<|k|\le n$. What remains to be shown is 
\begin{equation}
\| f_j^n(k)-\delta_{k,j}\|_{l^2(\Z)}=\sum_{k=-\infty}^{-n-1} \vert f_j^n(k)\vert^2+\sum_{k=n+1}^{\infty} \vert f_j^n(k)\vert^2\xrightarrow[n\to\infty]{} 0\label{equ:l2conv}
\end{equation}
for fixed $j$. For clarity of notation, we will prove this only for the contribution with $k>n$, as the statement for $k<-n$ follows analogously. 

Under this assumption, we can write \eqref{equ:interpoly} equivalently as
\begin{equation}
f_j^n(k)=(-1)^{n-j}\frac{\sinh(\lambda j)}{\sinh(\lambda k)}\frac{j}{(n-j)!(n+j)!}\frac{(k+n)!(k-n)}{(k-n)!k(k-j)},
\end{equation}
and for $n$ large enough we will have
\begin{equation}
|f_j^n(k)|<4j\sinh(\lambda j)g^n(k)\quad\text{with}\quad g^n(k)=\frac{(k+n)!(k-n)\E^{-\lambda k}}{(n!)^2(k-n)!k^2}.
\end{equation}

We will now derive an $l^\infty$-bound on $g^n(k)$. To this end we use Stirling's formula $m!=\sqrt{2\pi m}\big(\frac{m}{\E}\big)^m \big(1+O(\frac{1}{m})\big)$, which leads us to
\begin{equation}
g^n(k)=\frac{(k+n)^{k+n}}{(k-n)^{k-n}}\frac{\sqrt{k^2-n^2}}{2\pi k^2n^{2n+1}}\E^{-\lambda k}\big(1+O({\textstyle\frac{1}{k+n}})+O({\textstyle\frac{1}{k-n}})+O({\textstyle\frac{1}{n}})\big).\label{equ:gstirl}
\end{equation}
Note that as $k-n\ge 1$ and $n$ is positive (and large), the most critical asymptotic correction is $O(\frac{1}{k-n})$, stemming from the factor $(k-n)!$. Therefore, this corrections corresponds to a multiplication by a bounded function, and will only be relevant in a region where $f_j^n(k)$ is small anyway, as it has a zero crossing at $k=n$.

At this point, it is not hard to see that above function displays a scaling behaviour: defining the scaled independent variable $x=k/n$ leads to
\begin{equation}
\log\bigl(g^n(x)\bigr)\approx-2\log n+n\bigr(\log(x^2-1)+x\log{\textstyle\frac{x+1}{x-1}}-\lambda x\bigr)+{\textstyle\frac{1}{2}}\log(x^2-1)-2\log x-\log(2\pi).
\end{equation}
The asymptotic corrections can obviously only give a bounded additive contribution at $x\gtrsim 1$. As a consequence, the large $n$-behaviour of the $l^\infty$-norm of $g^n(k)$ is given by the sign of 
\begin{equation}
h(x)=\log(x^2-1)+x\log\Bigl(\frac{x+1}{x-1}\Bigr)-\lambda x.
\end{equation}

It can be directly seen that the unique maximum of $h(x)$ is attained at $x^*=\coth(\lambda/2)$, and that $h(x^*)=-2\log\bigr(\sinh(\lambda/2)\bigr)$. This result is independent of the sign of $j$, therefore by symmetry it holds also for $k<-n$. Hence we have proven that under the condition $\sinh(\lambda/2)\geq1$ (corresponding to $\lambda>1.763$ or $s<1.782$ since $s=\pi/\lambda$) $f_j^n(k)$ goes uniformly and exponentially to zero with large $n$. As it decays also exponentially in $k$, \eqref{equ:l2conv} follows.
\end{proof}

\section{Applications} \label{sec:appl}

This section is concerned with applications of Theorem~\ref{thm:result3}. First, we give a list of differential equations of the form~\eqref{equ:nonlinpde} for the first few choices of $\alpha$ and $\beta$. Then we use these equations for deriving explicit periodic solutions of the KdV equation and the Kawahara equation. Both represent nonlinear wave equations, where the latter contains additional higher order dispersive terms. In the case of the KdV equation we arrive at the well-known cnoidal wave solution derivable by direct integration~\cite{korteweg1895change} with one non-trivial free parameter. The Kawahara equation or fifth order KdV equation, however, is non-integrable. In this case we derive an explicit periodic travelling wave solution without non-trivial free parameter, in accordance with previous work~\cite{kano1981exact}.
With these two examples we demonstrate the procedure for finding periodic solutions to nonlinear differential equations, which are either exact or approximate solutions depending on the solvability of the system. 
More precisely, given a nonlinear ODE we make the ansatz
\begin{equation} \label{equ:ouransatz}
 f(\xi)=\sum_{i=0}^{2+\alpha+\beta}f_iu_s^{(i)}(\xi)+d,
\end{equation}
where $\alpha$ and $\beta$ are the orders of derivatives in the nonlinear term with the highest derivatives and two factors. In the case where the equation contains of higher multiplicative order one proceeds successively.

\subsection{The first few nonlinear differential identities} \label{subsec:specialcases}

We present certain nonlinear differential identities of the type~\eqref{equ:nonlinpde} for the periodic function $u_s$ given by Definition~\ref{def:us}. In~\eqref{equ:odes} we list specific examples of equations of low order. The coefficients are listed in Table~\ref{tab:coeff}.

\begin{equation} \label{equ:odes}
\begin{aligned}
&u_s\cdot u_s=b_{0,0}(2)u_s''+b_{0,0}(0)u_s+F_0-b_{0,0}(0)s/\pi,\\
&u_s'\cdot u_s=b_{1,0}(3)u_s'''+b_{1,0}(1)u_s',\\
&u_s'\cdot u_s'=b_{1,1}(4)u_s^{(4)}+b_{1,1}(0)u_s+F_2-b_{1,1}(0)s/\pi,\\
&u_s''\cdot u_s=b_{2,0}(4)u_s^{(4)}+b_{2,0}(2)u_s''+b_{2,0}(0)u_s-F_2-b_{2,0}(0)s/\pi,\\
&u_s''\cdot u_s'=b_{2,1}(5)u_s^{(5)}+b_{2,1}(1)u_s',\\
&u_s'''\cdot u_s=b_{3,0}(5)u_s^{(5)}+b_{3,0}(3)u_s'''+b_{3,0}(1)u_s',\\
&u_s''\cdot u_s''=b_{2,2}(6)u_s^{(6)}+b_{2,2}(2)u_s''+b_{2,2}(0)u_s+F_4-b_{2,2}(0)s/\pi,\\
&u_s'''\cdot u_s'=b_{3,1}(6)u_s^{(6)}+b_{3,1}(2)u_s''+b_{3,1}(0)u_s-F_4-b_{3,1}(0)s/\pi,\\
&u_s^{(4)}\cdot u_s=b_{4,0}(6)u_s^{(6)}+b_{4,0}(4)u_s^{(4)}+b_{4,0}(2)u_s''+b_{4,0}(0)u_s+F_4-b_{4,0}(0)s/\pi.
\end{aligned}
\end{equation}

\begin{table}[h!]
\begin{center}
\begin{tabular}{|c|c c c c c c c|}
\hline
&\multicolumn{7}{c|}{{\bf Values of the coefficients $b_{\alpha,\beta}(n)$}}\\
$n$	&	0	&	1	&	2	&	3	&	4	&	5	&	6	\\
\hline															
$b_{0,0}(n)$	&	$2(s/\pi-e_2)$	&	0	&	$-1/3$	&	0	&	0	&	0	&	0	\\
$b_{1,0}(n)$	&	0	&	$s/\pi-e_2$	&	0	&	$-1/6$	&	0	&	0	&	0	\\
$b_{1,1}(n)$	&	$-4e_4$	&	0	&	0	&	0	&	$-1/15$	&	0	&	0	\\
$b_{2,0}(n)$	&	$4e_4$	&	0	&	$s/\pi-e_2$	&	0	&	$-1/10$	&	0	&	0	\\
$b_{2,1}(n)$	&	0	&	$-2e_4$	&	0	&	0	&	0	&	$-1/30$	&	0	\\
$b_{3,0}(n)$	&	0	&	$6e_4$	&	0	&	$s/\pi-e_2$	&	0	&	$-1/15$	&	0	\\
$b_{2,2}(n)$	&	$-6e_6$	&	0	&	$2e_4$	&	0	&	0	&	0	&	$-1/70$	\\
$b_{3,1}(n)$	&	$6e_6$	&	0	&	$-4e_4$	&	0	&	0	&	0	&	$-2/105$	\\
$b_{4,0}(n)$	&	$-6e_6$	&	0	&	$10e_4$	&	0	&	$s/\pi-e_2$	&	0	&	$-1/21$	\\

\hline 																
\end{tabular}
\end{center}
\caption{Values of the coefficients $b_{\alpha,\beta}(n)$ defined by \eqref{equ:fouriercoeff1}--\eqref{equ:fouriercoeff2} for specific values of $\alpha$ and $\beta$. The quantities $e_{\ell}$ are defined by~\eqref{equ:eell}.}
\label{tab:coeff}
\end{table}

The coefficients of the equations~\eqref{equ:odes} contain the quantities $e_{\ell}$ defined by~\eqref{equ:eell}, that is
\begin{equation}  \label{equ:defeell}
e_{2n+2}=\frac{B_{2n+2}}{n+1}+4\sum_{k=1}^{\infty} \frac{k^{2n+1}}{1-\E^{2\pi k/s}},\qquad n\in \N_0, 
\end{equation}
and the infinite sum $F_{2n}$ given by \eqref{equ:infsumF}, that is
\begin{equation} \label{equ:deffell}
F_{2n}=\sum_{k\in \Z} \frac{k^{2n+2}}{\sinh^2(k\pi/s)},\qquad n\in \N_0.
\end{equation}
Note that $e_{2n+1}$ and $F_{2n+1}$ are zero for $n\in \N_0$. 

Above infinite sums show rapid convergence for small $s$. On the other hand, the exponential decay is slow for large $s$, which would necessitate to consider a large number of summands for reaching numerical precision. In this case, it is preferable to transform these expressions via Poisson summation, see Appendix~\ref{app:poissonsum}. More precisely, by using Lemma~\ref{thm:poissonspec} we can transform the expression~\eqref{equ:defeell} to read
\begin{equation} \label{equ:defeelllarge}
 e_{2n+2}=\frac{s}{\pi}\delta_{0,n}+(-1)^{n+1}s^2\Big( \frac{s^{2n}B_{2n+2}}{n+1}
-\frac{1}{(2\pi)^{2n}}\sum_{k=1}^{\infty} \frac{\partial^{2n}}{\partial k^{2n}}\frac{1}{\sinh^2(k\pi s)}\Big),
\end{equation}
where $n\in \N_0$.
Analogously, by applying Lemma~\ref{thm:poissonspec2} to~\eqref{equ:deffell} and using \eqref{equ:defeelllarge} we get
\begin{equation} \label{equ:deffelllarge}
 F_{2n}=\frac{2s^2}{\pi^2}\delta_{0,n}-(2n+2)e_{2n+2}\frac{s}{\pi}+(-1)^{n}\frac{2}{(2\pi)^{2n}}\frac{s^3}{\pi}
\sum_{k=1}^{\infty} \frac{\partial^{2n}}{\partial k^{2n}}\frac{k\pi s\cosh(k\pi s)}{\sinh^3(k\pi s)}.
\end{equation}

In the following, we will give the expressions for $e_{2n+2}$ and $F_{2n}$ with $n=0,1,2$ in both representations for purposes of reference. 

\subsection{Representation for small $s$}
The following expressions follow directly from \eqref{equ:defeell} and \eqref{equ:deffell}: 
\begin{equation} \label{equ:esmalls}
e_2=\frac{1}{6}+4\sum_{k=1}^{\infty} \frac{k}{1-\E^{2\pi k/s}}, \quad
e_4=-\frac{1}{60}+4\sum_{k=1}^{\infty} \frac{k^3}{1-\E^{2\pi k/s}}, \quad
e_6=\frac{1}{126}+4\sum_{k=1}^{\infty} \frac{k^5}{1-\E^{2\pi k/s}}
\end{equation}
\begin{equation} \label{equ:infsumFsmall}
F_0=\frac{s^2}{\pi^2}+2\sum_{k=1}^{\infty} \frac{k^2}{\sinh^2(\pi k/s)}, \quad
F_2=2\sum_{k=1}^{\infty} \frac{k^4}{\sinh^2(\pi k/s)}, \quad
F_4=2\sum_{k=1}^{\infty} \frac{k^6}{\sinh^2(\pi k/s)}
\end{equation}

\subsection{Representation for large $s$} \label{sec:repslarge}

The following expressions follow from \eqref{equ:defeelllarge} and \eqref{equ:deffelllarge}: 
\begin{equation} \label{equ:elarges}
\begin{aligned}
e_2&=\frac{s}{\pi}+s^2\Big\{ -\frac{1}{6}+\sum_{n=1}^{\infty}\frac{1}{\sinh ^2(n\pi s)}\Big\},\\
e_4&=s^4\Big\{ -\frac{1}{60}-\sum_{n=1}^{\infty}\frac{1}{\sinh ^2(n\pi s)}-\frac{3}{2}\sum_{n=1}^{\infty}\frac{1}{\sinh ^4(n\pi s)}\Big\},\\
e_6&=s^6\Big\{ -\frac{1}{126}+\sum_{n=1}^{\infty}\frac{1}{\sinh ^2(n\pi s)}+\frac{15}{2}\sum_{n=1}^{\infty}\frac{1}{\sinh ^4(n\pi s)}+\frac{15}{2}\sum_{n=1}^{\infty}\frac{1}{\sinh ^6(n\pi s)}\Big\}.
\end{aligned}
\end{equation}

\begin{equation} \label{equ:infsumFlarge}
\begin{aligned}
F_0=&-\frac{2s}{\pi}e_2+\frac{2s^2}{\pi^2}+\frac{s^3}{\pi}\sum_{n=1}^{\infty}\frac{2n\pi s\cosh(n\pi s)}{\sinh^3(n\pi s)},\\
F_2=&-\frac{4s}{\pi}e_4 -\frac{s^5}{\pi}\sum_{n=1}^{\infty}\frac{4n\pi s\cosh(n\pi s)+2n\pi s\cosh^3(n\pi s)}{\sinh^5(n\pi s)},\\
F_4=&-\frac{6s}{\pi}e_6 +\frac{s^7}{\pi}\sum_{n=1}^{\infty}\frac{17n\pi s\cosh(n\pi s)+26n\pi s\cosh^3(n\pi s)+2n\pi s\cosh^5(n\pi s)}{\sinh^7(n\pi s)}.
\end{aligned}
\end{equation}
Note that by the calculation of the coefficients $c_{\alpha,\beta}$ defined by~\eqref{equ:constcoeff} the first terms on the right hand side of~\eqref{equ:infsumFlarge} cancel out.

\subsection{The KdV equation}

The Korteweg--de Vries (KdV) equation is given by 
\begin{equation} \label{equ:KdV}
v_t+vv_z+\alpha v_{zzz}=0,
\end{equation}
where $\alpha$ is a constant. It is perhaps the most famous example for an integrable nonlinear wave equation. Soliton as well as periodic solutions have already been obtained by Korteweg and de~Vries~\cite{korteweg1895change}. Thanks to its integrability there exist algebro-geometric methods to derive more general (quasi-)periodic finite-gap solutions for this equation~\cite{gesztesy2003sea}.

In the following we will show how to derive the well-known periodic cnoidal wave solution by applying our method. Inserting the travelling wave ansatz $v(z,t)=f(x)$, with $x=z-ct$, and integrating once we get
\begin{equation} \label{equ:KdVint}
f^2=-2\alpha f_{xx}+2cf+d,
\end{equation}
where $d$ denotes an arbitrary integration constant. In the following, we determine the solution of this equation by using Theorem~\ref{thm:result3}. More precisely, we make the ansatz
\begin{equation} \label{equ:ansatzkdv}
f(x)=f_1u_s(x),
\end{equation}
where $f_1$ is a constant to be determined and $u_s$ denotes the $2\pi$-periodic function given by Definition~\ref{def:us}.
Inserting this ansatz into~\eqref{equ:KdVint} and using the relation~\eqref{equ:odes} for $u_s^2$ finding the solution reduces to the problem of comparing coefficients. We get
\begin{equation*}
f_1=6\alpha,\quad 
c=-6\alpha(e_2-\frac{s}{\pi}),\quad
d=36\alpha^2\big( F_0+2\frac{s}{\pi}( e_2-\frac{s}{\pi})\big),
\end{equation*}
where explicit expressions for $e_2$ and $F_0$ are given in~\eqref{equ:esmalls}--\eqref{equ:infsumFlarge}. Hence the function $f$ given by~\eqref{equ:ansatzkdv} with the coefficients determined as above solves~\eqref{equ:KdVint}. In particular, a $2\pi$-periodic solution of~\eqref{equ:KdV} is given by
\begin{equation} \label{equ:KdVsol}
v(z,t)=6\alpha u_s(z-ct),\qquad c=\alpha s^2\Big(1-6\sum_{n=1}^{\infty}\frac{1}{\sinh ^2(n\pi s)}\Big),
\end{equation}
where $u_s$ is the $2\pi$-periodic function given by Definition~\ref{def:us}. Here we used the representation~\eqref{equ:elarges} for $e_2$. Note that the solution~\eqref{equ:KdVsol} still depends on one free parameter $s\in(0,\infty)$. Two additional degrees of freedom correspond to shifting or scaling the solution:
\begin{equation*}
 v(z,t)=a+6\alpha \lambda^2u_s\big(\lambda z-\lambda^3(c-a)t\big), \qquad a,\lambda\in \R.
\end{equation*}

\subsection{The Kawahara equation}

The Kawahara equation, or higher order KdV equation, is given by
\begin{equation} \label{equ:HKdV}
v_t+vv_z+\alpha v_{3z}-\beta v_{5z}=0,
\end{equation}
where $\alpha$ and $\beta$ are constants and $v_{kz}$ denotes the $k$th derivative of $v$ with respect to $z$. The equation was derived by Kakutani and Ono~\cite{kakutani1969weak} in the context of studying weakly nonlinear magneto-acoustic waves in a cold collisionless plasma. A numerical investigation by Kawahara~\cite{kawahara1972oscillatory} established the existence of solitary wave solutions. Although equation~\eqref{equ:HKdV} is in general non-integrable, analytic expressions for exact solitary and periodic solutions have been found by Kano and Nakayama \cite{kano1981exact}, using a polynomial ansatz of the form~\eqref{equ:ansatzofothers} in terms of the Weierstra\ss{} function $\wp$, and by Yamamoto and Takizawa~\cite{yamamoto1981solution} by direct integration.

In the following we apply our method to derive an exact periodic travelling wave solution of~\eqref{equ:HKdV}.
Inserting the travelling wave ansatz $v(z,t)=f(x)$, with $x=z-ct$, and integrating once we get
\begin{equation} \label{equ:HKdVint}
f^2=2\beta f_{4x}-2\alpha f_{2x}+2cf+d,
\end{equation}
where $d$ denotes an arbitrary integration constant. Writing
\begin{equation} \label{equ:ansatzhkdv}
f(x)=f_1u_s(x)+f_2u_s''(x),
\end{equation}
where $f_1$ and $f_2$ are constants to be determined and $u_s$ denotes the $2\pi$-periodic function given by Definition~\ref{def:us}. Inserting this ansatz into~\eqref{equ:HKdVint} and using the relations~\eqref{equ:odes} for $u_s^2$, $u_s''u_s$, and $(u_s'')^2$ leads again to a problem of comparing the coefficients. The results are
\begin{equation} \label{equ:systemkaw}
\begin{aligned}
\beta f_2=&-\frac{1}{140}f_2^2,\\
\beta f_1+\alpha f_2=&-\frac{1}{10}f_1f_2,\\
\alpha f_1+cf_2=&-\frac{1}{6}f_1^2-f_1f_2(e_2-\frac{s}{\pi})+f_2^2e_4,\\
cf_1=&-f_1^2(e_2-\frac{s}{\pi})+4f_1f_2e_4-3f_2^2e_6,\\
d=&2f_1^2\frac{s}{\pi}(e_2-\frac{s}{\pi})-8f_1f_2\frac{s}{\pi}e_4+6f_2^2\frac{s}{\pi}e_6
				+f_1^2F_0-2f_1f_2F_2+f_2^2F_4,
\end{aligned}
\end{equation}
where the quantities $e_{\ell}$ and $F_{\ell}$ are defined in \eqref{equ:defeell} and \eqref{equ:deffell}, respectively.
We seek non-trivial solutions of the form~\eqref{equ:ansatzhkdv}, i.e. $f_1$ and $f_2$ not both equal zero, by solving the system~\eqref{equ:systemkaw}. We get from the first equation
\begin{equation*}
f_2=-140\beta.
\end{equation*}
Then if $\beta\neq 0$ we get from the second equation
\begin{equation*}
f_1=\frac{140}{13}\alpha.
\end{equation*}
Otherwise, if $\beta=0$, the problem reduces to the KdV case.

Combining the third and the fourth equation of system~\eqref{equ:systemkaw} leads to the constraint
\begin{equation} \label{equ:hilffunkt}
g(s, \alpha,\beta)=0 \quad \text{with $g(s, \alpha,\beta)=31\alpha^3+212940\alpha\beta^2e_4(s)+2768220\beta^3e_6(s)$}.
\end{equation}

For fixed $\alpha$ and $\beta\neq 0$ let $s_0>0$ be such that $g(s_0,\alpha,\beta)=0$. Then $c$ is determined by 
\begin{equation} \label{equ:kawaharac}
 c=\frac{31 \alpha^2}{507 \beta}-\frac{140}{13} \alpha \big( e_2(s_0)-\frac{s_0}{\pi}\big) -140 \beta e_4(s_0).
\end{equation}

In the following we will determine a region in the $(\alpha,\beta)$-plane for which the equation~\eqref{equ:hilffunkt} has a solution $s_0$.
For that purpose note that a point $(\alpha, \beta)$ lies in this region if and only if $(\lambda\alpha, \lambda\beta)$ does for all $\lambda\in \R$. Therefore it suffices to consider a fixed value $\beta\in \R\backslash \{0\}$, say $\beta=1$. 

Moreover, note that the quantities $e_4(s)$ and $e_6(s)$ fulfill the following properties: first, by the representations~\eqref{equ:esmalls} for small $s$, we have $e_4(0^+)=-1/60$ and $e_6(0^+)=1/126$, and second, by~\eqref{equ:elarges} for large $s$, we get $e_4(s)\sim -s^4/60$ and $e_6(s)\sim -s^6/126$. 

Next set $s=0$, then the equation $g(0,\alpha,1)=0$ has the real solution $\alpha_0=-13$. We have $g(0,\alpha<\alpha_0,1)<0$ and $g(0,\alpha>\alpha_0,1)>0$. Since $g(s,\alpha,1)\to -\infty$ as $s\to \infty$ we finally get that for all $\alpha>\alpha_0$ there exists a $s_0\in (0,\infty)$ such that $g(s_0,\alpha,1)=0$. This $s_0$ is unique if $\alpha>0$ since in this case $g$ is monotone in $s$.

Thus, finally, we deduce that for $(\alpha,\beta)\in \Gamma$ with
\begin{equation*}
 \Gamma=\{ (\alpha,\beta)\in \R\times \R \backslash\{0\}: \alpha/\beta>-13\},
\end{equation*}
equation~\eqref{equ:HKdV} has a periodic solution given by
\begin{equation} \label{equ:solkawahara}
v(z,t)=\frac{140}{13}\alpha u_{s_0}(z-ct)-140\beta u_{s_0}''(z-ct),
\end{equation}
where $u_{s_0}$ is the $2\pi$-periodic function given by Definition~\ref{def:us}. The velocity $c$ is determined by~\eqref{equ:kawaharac} and $s_0$ is a solution of equation~\eqref{equ:hilffunkt}.

For instance, in the case $\alpha=-1$ and $\beta=1$ we numerically determined $s_0=1.0346$ such that $g(s_0,-1,1)=0$ and the velocity computed from~\eqref{equ:kawaharac} is $c=1.8602$.

The solution contains no free parameter since $s_0$ and $c$ are fixed. Note, however, that -- as outlined above for the  KdV case -- one can still add a trivial degree of freedom by shifting the solution by an additive constant. However, the degree of freedom corresponding to scaling by a parameter $\lambda$ is best realised by substituting $\alpha\to \lambda^{-2}\alpha$ and $\beta\to \lambda^{-4}\beta$.

In the case $\alpha=0$ the solution~\eqref{equ:solkawahara} coincides with the one derived by Kano and Nakayama~\cite{kano1981exact} in terms of the Weierstra\ss\ elliptic function $\wp$ (the solution in~\cite{kano1981exact}, however, contains an error, it is valid except for a factor $1/2$, as already pointed out in~\cite{parkes2002jacobi}). As in our case the solution in~\cite{kano1981exact} contains no free (non-trivial) parameter and is exactly~\eqref{equ:solkawahara} with $m=0.5$ ($s=1$).

\appendix

\section{Poisson summation} \label{app:poissonsum}

This appendix is primarily based on \cite[Section~3-3]{papoulis1962fourier} (besides we also refer to~\cite[Section~7.2]{stein1975introduction}) and should serve as a suitable reference for Sections~\ref{sec:intro}~and~~\ref{sec:appl}, where the idea of Poisson summation is used. 

Let $f$ be an arbitrary function that can be expressed in the form
\begin{equation*}
f(x)=\frac{1}{2\pi}\int_{-\infty}^{\infty} F(\omega)\E^{\I \omega x}d\omega,
\end{equation*}
Here the function $F$ denotes the \emph{Fourier transform} of $f$ and is given by
\begin{equation} \label{equ:FourierTF}
F(\omega)=\int_{-\infty}^{\infty} f(x)\E^{-\I \omega x}dx.
\end{equation}

The following identity is known as \emph{Poisson's sum formula}:
\begin{equation} \label{equ:poissonsum}
\sum_{k=-\infty}^{\infty} f(x+kT)=\frac{1}{T} \sum_{k=-\infty}^{\infty}F\big(\frac{2\pi k}{T}\big) \E^{\frac{2\pi \I kx}{T}}.
\end{equation}
Setting $x=0$ in \eqref{equ:poissonsum} we immediately get the special form
\begin{equation} \label{equ:poissonspec}
\sum_{k=-\infty}^{\infty} f(kT)=\frac{1}{T} \sum_{k=-\infty}^{\infty}F\big(\frac{2\pi k}{T}\big).
\end{equation}
Note that if the function $f$ is discontinuous at $kT$, its value at this point has to be taken as $(f(kT^+)+f(kT^-))/2$. For our purposes, let us consider a function $f(x)$ which is discontinuous at $x=0$ and $f(x)=0$ for all $x<0$. Then \eqref{equ:poissonspec} becomes
\begin{equation*}
\frac{f(0^+)}{2}+\sum_{k=1}^{\infty}f(nT)=\frac{1}{T} \sum_{k=-\infty}^{\infty}F\big(\frac{2\pi k}{T}\big),
\end{equation*}
or equivalently
\begin{equation} \label{equ:poissonhalf}
\frac{f(0^+)}{2}+\sum_{k=1}^{\infty}f(kT)=\frac{1}{T} \int_0^{\infty} f(x)dx+2\sum_{k=1}^{\infty}\int_0^{\infty} f(x)\cos(2\pi kx)dx.
\end{equation}
We refer to \eqref{equ:poissonhalf} as \emph{Poisson summation on the half-line}.

As an application of~\eqref{equ:poissonhalf} we have the following two results, which will be useful in Section~\ref{sec:repslarge}.

\begin{lemma} \label{thm:poissonspec}
Let $s\in (0,\infty)$. Then we have
\begin{equation*}
\sum_{k=1}^{\infty}\frac{k^{2n+1}}{1-\E^{2\pi k/s}}=\frac{s}{4\pi}\delta_{0,n}+\Big( \frac{B_{2n+2}}{4(n+1)}
\big( (-1)^{n+1}s^{2n+2}-1 \big) +(-1)^n\frac{s^2}{4(2\pi)^{2n}}\sum_{k=1}^{\infty} \frac{\partial^{2n}}{\partial k^{2n}}\frac{1}{\sinh^2(k\pi s)}\Big),
\end{equation*}
where $\delta_{i,j}$ denotes the Kronecker delta, i.e.~$\delta_{i,j}=1$ for $i=j$ and $\delta_{i,j}=0$ else, and $B_n$ denotes the $n$th Bernoulli number, see Remark~\ref{rem:bernoullinumbers}.
\end{lemma}
\begin{proof}
We will apply formula~\eqref{equ:poissonhalf} to the real-valued function
\begin{equation*}
f_n(x)=x^{2n+1}/(1-\E^{2\pi x/s})\quad \text{for $n\in\N_0$}. 
\end{equation*}
For that purpose note that 
\begin{equation*}
\int_0^{\infty} f_0(x)\cos(2\pi kx)dx=-\frac{1}{8\pi^2}\frac{1}{k^2}+\frac{s^2}{8}\frac{1}{\sinh^2(k\pi s)}
\end{equation*}
and by differentiating $2n$ times with respect to $k$ we get
\begin{equation} \label{equ:eqproof1}
\int_0^{\infty} f_n(x)\cos(2\pi kx)dx=(-1)^{n+1}\frac{1}{(2\pi)^{2n}}\Big( \frac{1}{8\pi^2}\frac{(2n+1)!}{k^{2n+2}}-\frac{s^2}{8} \frac{\partial^{2n}}{\partial k^{2n}}\frac{1}{\sinh^2(k\pi s)}\Big).
\end{equation}
Moreover, we have
\begin{equation} \label{equ:eqproof2}
\int_0^{\infty}f_n(x)dx=(-1)^{n+1}\frac{s^{2n+2}}{2(2n+2)}B_{2n+2}.
\end{equation}
Thus, inserting~\eqref{equ:eqproof1} and~\eqref{equ:eqproof2} into~\eqref{equ:poissonhalf} and additionally using that
\begin{equation*} 
\sum_{k=1}^{\infty}\frac{1}{k^{2n+2}}=(-1)^{n}\frac{(2\pi)^{2n+2}}{2(2n+2)!}B_{2n+2},
\end{equation*}
proves the stated expression.
\end{proof}

\begin{lemma} \label{thm:poissonspec2}
Let $s\in (0,\infty)$. Then we have
\begin{equation*}
\sum_{k=1}^{\infty}\frac{k^{2n+2}}{\sinh^2(k\pi /s)}=(-1)^n \frac{s^{3}}{\pi}\Big( s^{2n}B_{2n+2}+\frac{1}{(2\pi)^{2n}}\sum_{k=1}^{\infty} \frac{\partial^{2n}}{\partial k^{2n}}\big(\frac{k\pi s\cosh(k \pi s)-\sinh(k\pi s)}{\sinh^3(k\pi s)}\big)\Big),
\end{equation*}
where $B_n$ denotes the $n$th Bernoulli number, see Remark~\ref{rem:bernoullinumbers}.
\end{lemma}
\begin{proof}
We will apply formula~\eqref{equ:poissonhalf} to the real-valued function
\begin{equation*}
f_n(x)=x^{2n+2}/\sinh^2(x\pi/s)\quad \text{for $n\in\N_0$}. 
\end{equation*}
For that purpose note that
\begin{equation} 
\int_0^{\infty}f_n(x)dx=(-1)^{n}\frac{s^{2n+3}}{\pi}B_{2n+2}.
\end{equation}
Moreover, we have
\begin{equation*}
\int_0^{\infty} f_0(x)\cos(2\pi kx)dx=\frac{s^3(1-k\pi s\coth(k\pi s))}{2\pi\sinh^2(k\pi s)}.
\end{equation*}
Differentiating the last equation $2n$ times with respect to $k$ and applying~\eqref{equ:poissonhalf} finishes the proof.
\end{proof}

 \noindent

\vspace{1cm}

{\bf Acknowledgments.}

A.M.~gratefully acknowledges the financial support by the Austrian Science Fund (FWF):~J3143.


\end{document}